\documentclass[12pt]{article}
\usepackage{amsmath}
\usepackage{eurosym}
\usepackage{amsfonts}
\usepackage[onehalfspacing]{setspace}
\usepackage{geometry}
\usepackage{amsthm}
\usepackage{indentfirst}
\usepackage{titlesec}
\usepackage{xcolor}
\usepackage[sort,comma,authoryear,round]{natbib}
\usepackage{tikz}
\usepackage[breaklinks,colorlinks, citecolor=blue, linkcolor=black, urlcolor=blue,pdfsubject={LaTeX}]{hyperref}

\definecolor{darkred}{rgb}{0.7,0,0}
\definecolor{darkblue}{rgb}{0,0,0.7}
\usetikzlibrary{arrows}
\usetikzlibrary {patterns}
 \titleformat{\section}
  {\normalfont\fontsize{14pt}{14pt}\bfseries }{\thesection}{1em}{}
\titleformat{\subsection}
  {\normalfont\fontsize{12pt}{12pt}\bfseries }{\thesubsection}{1em}{}
\let\oldref\ref
\renewcommand{\ref}[1]{\textup{(\oldref{#1})}}
\newtheorem{cor}{Corollary}
\newtheorem{obs}{Observation}
\newtheorem{prop}{Proposition}
\newtheorem{df}{Definition}
\newtheorem{ex}{Example}

\newtheorem{tm}{Theorem}

 \input{tcilatex}
\begin{document}

\title{ Two-Person Bargaining when the Disagreement Point is Private
Information\thanks{
We thank the Advisory Editor and two anonymous referees for their comments.
The present paper is a substantially revised and extended version of the
paper \textquotedblleft How to compromise when compromising is inefficient".
We thank conference participants at EEA-ESEM 2018, Nanjing ICGT 2018 and
ESWC 2020 for helpful discussions.}}
\author{Eric van Damme\thanks{
Department of Economics, Tilburg University; Eric.vanDamme@uvt.nl.} \and Xu
Lang\thanks{%
Department of Economics, Southwestern University of Finance and Economics;
langxu1011@gmail.com.} }
\date{January 9, 2024 }
\maketitle

\begin{abstract}
We consider two-person bargaining problems in which (only) the disagreement
outcome is private (and possibly correlated) information and it is common
knowledge that disagreement is inefficient. We show that if the Pareto
frontier is linear, the outcome of an \textit{ex-post} efficient mechanism
cannot depend on the disagreement payoffs. If the frontier is non-linear,
the result continues to hold when the disagreement payoffs are independent
or there is a player with at most two types. We discuss implications of
these results for axiomatic bargaining theory and for full surplus
extraction in mechanism design. \bigskip


\end{abstract}

\thispagestyle{empty}\setcounter{page}{0}%
\renewcommand{\thefootnote}{\arabic{footnote}}%
\newpage

\section{Introduction}

In this paper, we analyze two-person bargaining problems in which the
players have private information \textit{only} about the disagreement
outcome and in which it is common knowledge that Pareto efficiency requires
agreement to be reached. Such a situation could, for example, arise when the
players know each other very well, but each is lacking information about the
other's outside options. We study whether, in such situations, these outside
options will influence the bargaining outcome.

As a simple example, suppose two players have to allocate an (indivisible)
dollar to one of them, with each player only getting 25 cents if they fail
to agree. Intuition suggests that the outcome will depend on the players'
risk attitudes and, indeed, under complete information, all standard
bargaining solutions predict that a more risk averse player will get the
dollar with a \textit{larger} probability. However, risk attitudes typically
are private information. Our results imply that, in this case, if it is
common knowledge that both players are not too risk averse (and, for
example, prefer the 50/50 lottery to disagreement), the degree of risk
aversion is irrelevant: any symmetric and \textit{ex-post} efficient
mechanism has to prescribe the 50/50 lottery irrespective of the players'
risk attitudes. More generally, we show that, for bargaining problems with a
linear Pareto frontier or independent disagreement payoffs, any \textit{%
ex-post} efficient mechanism has to be `disagreement point independent',
i.e., the players' (interim) expected utilities cannot depend on the values
of their outside options.

The situation that we analyze differs from the classic bilateral trade model
of \cite{MS83}. In that model, disagreement (no trade) \textit{can be}
efficient, and there does not exist an \textit{ex-post} efficient, incentive
compatible and individually rational mechanism. In our context, disagreement
is inefficient and -as in the example above- the players may even know a
lottery that Pareto dominates disagreement for sure. Furthermore, there
typically are many mechanisms that satisfy these three conditions; in the
example, any (constant) mechanism that always selects the same lottery close
to 50/50 satisfies them. The main result of the present paper is that only
such `constant' solutions can satisfy them: `disagreement point
independence' holds generically as long as the Pareto frontier is linear.

The model that we analyze amounts to a minimal change of the canonical
(Nash) bargaining game with complete information. In that context, it is
usually assumed that the disagreement point is inefficient, and Pareto
efficiency of the outcome is a standard axiom. The major bargaining
solutions that have been proposed for this situation all satisfy \textit{%
Disagreement Point Monotonicity} (DPM): if a player's disagreement payoff
increases, then the outcome (weakly) moves in her favor \citep{TH87}.%
\footnote{%
DPM should be distinguished from Monotonicity (M), the main axiom that
characterizes proportional solutions \citep{KA77} and which insists that, if
the set of possible agreements is enlarged and the disagreement outcome does
not change, none of the players becomes worse off.} DPM is an intuitive idea
and is a building block of the Harvard Negotiation Model \citep{FU81}. The
fifth principle of this negotiation method states that a bargainer should
enhance her Best Alternative to Negotiated Agreement (BATNA):
\textquotedblleft The better your BATNA, the greater your power (...). In
fact, the relative negotiating power of the two parties depends primarily
upon how attractive to each is the option of not reaching an
agreement.\textquotedblright (\textit{Getting to Yes, }p. 52\textit{).}%
\footnote{%
Of course, BATNA can only be satisfied in a weak sense: a player's
bargaining power does not decrease, but will sometimes increase, if her
outside option improves. The outside option principle \citep{BRW86} holds
that non-credible options are irrelevant.}

With complete information, efficiency and DPM, hence, are compatible. We
only change the canonical model by making the players' disagreement payoffs
private information and study whether efficiency is compatible with the
bargaining outcome being responsive to a player's outside option. Although
such responsiveness (called type-dependence below) is a weaker property than
monotonicity, we show that it conflicts with efficiency. Hence, in the\
`allocate the dollar example', when the players' risk attitudes are private
information, they cannot influence the lottery that the players will agree
upon.\footnote{%
In this example, type-independence corresponds to ordinality: the mechanism
only depends on the players' ordinal preferences (each prefers more money to
less), but not on the intensity of the players' preferences, i.e., how much
each dislikes disagreement.} Consequently, the incomplete information case
is very different from the one with complete information. BATNA does not
hold in this case. Although DPM is a very weak requirement in the latter
domain, type-dependence is very demanding in the former.

We admit that type-dependence is somewhat less appealing than DPM is under
complete information. If the disagreement payoff is private information, a
player can \textit{claim} that her outside option has improved, but, as the
other party may not be able to \textit{verify} this, why should he respond
to it? The claim might be ignored, in particular, if the players' outside
options are independent or positively correlated. However, if they are
negatively correlated, the second player can conclude that his partner's
claim is likely to be true, and, equally important, that she knows that he
knows this. Hence, under negative correlation, type-dependence has intuitive
appeal: a player with a better option might be a tougher bargainer and this
could translate into a different outcome. This certainly holds when the
disagreement utilities are perfectly negatively correlated as this case is
equivalent to having complete information. Hence, our intuition is that,
with negative correlation, type-dependence and efficiency can be satisfied
simultaneously. We, however, show that this intuition is wrong; our main
result is that if the Pareto frontier is linear (and the players' prior has
full support), any feasible, efficient mechanism must be \textit{%
type-independent}: how much a player dislikes disagreement cannot play any
role. While this result may not be surprising when the players' disagreement
utilities are independent, it is far from obvious that it holds for all
correlated priors.

Our{\ paper aims to contribute to the literature on axiomatic bargaining
with incomplete information. There is only a small literature on this topic.
Two solution concepts have been proposed and axiomatized: the generalized
Nash bargaining solution (\citealp{HS72}; \citealp{MY79}) and the neutral
bargaining solution \citep{MY84}. Occasionally also the \textit{ex-ante}
utilitarian solution (which }maximizes the sum of \textit{ex-ante} expected
utilities of the players over the set of all feasible mechanisms) is used,
as in \cite{MS83}. It can be shown that, in the above `allocate the dollar'
example, when the players' disagreement utilities are independent, the two
former concepts are type-dependent (and, hence, predict that there will be
some inefficiency), while the latter is efficient and, hence,
type-independent. Our results imply that no solution can be both efficient
and type-dependent.

This impossibility result raises the question of which of these properties,
efficiency or type-dependence, should get priority. One might argue that,
given that efficient solutions do not exist in general, efficiency should be
given up. We, however, think that the issue is not that clear: if players
cannot commit to a mechanism, then an inefficient mechanism is not
renegotiation-proof. The players might not agree to a mechanism \textit{%
ex-ante} if they know that they will renegotiate \textit{ex-post} once they
see that the mechanism's outcome is inefficient.

The remainder of the paper is organized as follows. Section 2 introduces our
main model. In Section 3 we show that, if the Pareto frontier is linear (and
the prior has full support), an efficient mechanism must be
type-independent. Section 4 shows that this result no longer holds when the
frontier is non-linear. In Section 5, we adapt our model to the transferable
utility case and show that the results from the Sections 3 and 4 continue to
hold. Section 6 discusses three applications of our results for bargaining
theory and Section 7 concludes. The more technical proofs are provided in
the Appendix.

\section{The model}

We consider a two-person bargaining problem $(A,a_{0})$, where $A$ is a
finite set of alternatives with $|A|\geq 3$ and $a_{0}\in A$ is the
disagreement (or status-quo) outcome. We assume that the\ players have Von
Neumann Morgenstern utility functions $u_{1}$ and $u_{2}$, denote by $%
u_{i}^{a}$ the utility that $i$ attaches to $a\in A$ and write $%
u^{a}=(u_{1}^{a},u_{2}^{a})$. We write $A^{0}=A\setminus \{a_{0}\}$ for the
set of (real) agreements and assume that the utilities of all agreements, $%
\{u^{a}|a\in A^{0}\}$, are common knowledge, but that $u_{i}^{0}:=$ $%
u_{i}^{a_{0}}$ is private information of player $i$. Without loss of
generality, we assume that different agreements have different utilities: if
$a,a^{\prime }\in A^{0}$ and $a\neq a^{\prime }$, then $u_{i}^{a}\neq
u_{i}^{a^{\prime }}$ for $i=1,2$. Writing $a_{i}$ for the best Pareto
efficient alternative for player $i$ from $A^{0}$, we thus have $a_{1}\neq
a_{2}$ and $a_{i}$ also is the worst Pareto efficient alternative from $%
A^{0} $ for player $j$. We write $u^{1}:=$ $u^{a_{1}}$ and $u^{2}:=$ $%
u^{a_{2}}$. We denote by $U^{0}=conv\{u^{a}|a\in A^{0}\}$ the set of all
utility pairs resulting from an agreement (where `$conv$' denotes the convex
hull) and by $\partial U^{0}$ the Pareto boundary of $U^{0}${; hence, }$%
\partial U^{0\text{ }}$is the set of all $u\in U^{0}$ for which there does
not exist $v\in U^{0}$ with $v_{1}\geq u_{1}$ and $v_{2}\geq u_{2}$, with at
least one of these inequalities being strict. (Note that, since there are no
indifferences, weak and strong Pareto efficiency coincide.)

The model reflects a real-world scenario in which the payoffs from all
possible agreements are common knowledge, but a player does not know how
much the other one dislikes disagreement. Such a situation could, for
example, arise when the players know each other very well, but each is
lacking information about the other's outside options.

As $u_{i}^{0}$ (the utility of disagreement) is private information of
player $i$, the possible values of $u_{i}^{0}$ constitute the types of this
player. For each player $i$, let $T_{i}$ denote the (finite) set of his
possible types and let $T=T_{1}\times T_{2}$ denote the product type set
with generic element $t=(t_{1},t_{2})$. Denote $\bar{t}_{i}:=\max T_{i}$ and
$\underline{t}_{i}:=\min T_{i}$. The players have a common prior $f$ on $T$,
with $f_{i}(t_{i})=\sum_{t_{j}\in T_{j}}f(t_{i},t_{j})>0$ for $i=1,2$ and
all $t_{i}\in T_{i}$. We write $f_{i}^{c}(t_{j}|t_{i})=f(t)/f_{i}(t_{i})$
for the associated conditional probability. We say that $f$ has full support
if $f(t)>0$ for all $t\in T$. Note that we allow the players' types to be
correlated. We write $\Gamma =(A,a_{0},T_{1},T_{2},u_{1},u_{2},f)$ for this
\textit{Bayesian bargaining problem}. As is standard in the literature %
\citep{MY79,MY84}, we assume that each player already knows his type when
the bargaining starts.

As far as disagreement is concerned, we make two assumptions. First, $%
u_{i}^{j}<\underline{t}_{i}<\bar{t}_{i}<u_{i}^{i}$ for each player $i$ .
This implies that player $j$'s best outcome is unacceptable to player $i$ ($%
i\neq j$) so that the players need to compromise. Second, we assume that it
is common knowledge that $a_{0}$ is Pareto inefficient, i.e., for every $%
t\in T$, there exists a lottery $\alpha (t)\in \Delta (A^{0})$ that Pareto
dominates $a_{0}$. Hence, if we write $U=conv[U^{0}\cup \{t|t\in T\}]$ and $%
\partial U$ for its Pareto frontier, we have $\partial U=\partial U^{0}$.
Note that, if $f$ has full support, it is common knowledge that the lottery $%
\alpha (\bar{t}_{1},\bar{t}_{2})$ is better than disagreement. When $f$ does
not have full support, the set of dominating lotteries can be contingent on $%
t$, with no lottery dominating all $t\in T$, as illustrated by case (b) in
Figure 1 below. This figure depicts situations in which the players' utility
functions have been normalized such that $u^{1}=(1,0)$ and $u^{2}=(0,1)$.
Such normalization is without loss of generality and, whenever convenient,
we will use it. In a normalized problem, $\partial U$ is a piece-wise linear
curve that connects $(0,1)$ to $(1,0)$. Figure 1 shows three cases of our
model: (a) a linear frontier with full support; (b) a linear frontier with a
triangular support; and (c) a non-linear frontier with full support. The
disagreement payoffs are distributed in the dotted regions.

\begin{tikzpicture}[
    scale=3.5,
    axis/.style={very thick, ->, >=stealth'},
    important line/.style={thick},
    dashed line/.style={dashed, thin},
    pile/.style={thick, ->, >=stealth', shorten <=2pt, shorten
    >=2pt},
    every node/.style={color=black}
    ]
    \draw[axis] (0,0)  -- (1.2,0) node(xline)[right]
        {$u_2$};
    \draw[axis] (0,0) -- (0,1.2) node(yline)[above] {$u_1$};
         \draw[important line] (0,1) coordinate (B1) circle
         (.4pt) node[left, text width=1em] {$1$} -- (1,0)
         coordinate (A1) circle (.4pt)     node[below, text
         width=0em] {$1$};

    \fill[green](A1) circle (.6pt)  ;
    \fill[green]  (B1) circle (.6pt)  ;
 \draw[pattern=dots] (0,0)node[below, text width=1em] {$O$}
 rectangle (0.3,0.5);

 \draw (0.4,1.4) node[below, text width=1em] {$(a)$};

     \draw[axis] (1.5,0)  -- (2.7,0) node(xline)[right]
        {$u_2$};
    \draw[axis] (1.5,0) -- (1.5,1.2) node(yline)[above] {$u_1$};

    \draw[important line] (1.5,1) coordinate (B2) circle (.4pt)
    node[left, text width=1em] {$1$} -- (2.5,0)  coordinate (A2)
    circle (.4pt)     node[below, text width=0em] {$1$};

    \fill[green](A2) circle (.6pt)  ;
    \fill[green]  (B2) circle (.6pt)  ;
 \draw[pattern=dots] (1.5,0)node[below, text width=1em] {$O$} --
 (2.4,0)--(1.5,0.9);

  \draw (1.9,1.4) node[below, text width=1em] {$(b)$};


      \draw[axis] (3,0)  -- (4.2,0) node(xline)[right]
        {$u_2$};
    \draw[axis] (3,0) -- (3,1.2) node(yline)[above] {$u_1$};

    \draw[important line] (3,1) coordinate (B3) circle (.4pt)
    node[left, text width=1em] {$1$} --     (3.8, 0.6)
    coordinate (C3) circle (.4pt)     node[below, text
    width=0em] { }-- (4,0)  coordinate (A3) circle (.4pt)
    node[below, text width=0em] {$1$};

    \fill[green](A3) circle (.6pt)  ;
    \fill[green]  (B3) circle (.6pt)  ;
  \fill[green]  (C3) circle (.6pt)  ;

 \draw[pattern=dots] (3,0)node[below, text width=1em] {$O$}
 rectangle (3.3,0.5);

   \draw (3.4,1.4) node[below, text width=1em] {$(c)$};

     \end{tikzpicture}

\begin{center}
Figure 1
\end{center}

Our model amounts to a minimal change of the canonical (Nash) bargaining
game with complete information; the only difference is that in our model the
disagreement utilities are private information. Our model is {inspired by
\cite{BP09}, in which the\ case with} $A=\{a_{0},a_{1},a_{2}\}$ and
independent, symmetric beliefs is investigated. That paper differs in its
assumptions and emphasis. It allows $a_{0}$ to be efficient (and, hence, to
be a real compromise), shows that this implies that no efficient mechanism
can be incentive compatible and derives some numerical properties of \textit{%
ex-ante} efficient incentive compatible rules. Our assumption that $a_{0}$
is inefficient makes the model more tractable. On the other hand, our model
is more general since we allow correlated (and asymmetric) beliefs and take
into account individual rationality constraints.

\subsection{Mechanisms and desirable properties}

Let a Bayesian bargaining game $\Gamma $ be given. Because of the revelation
principle \citep{MY79}, we can restrict ourselves to (incentive compatible)
direct mechanisms in which the message spaces are the type spaces. A \textit{%
mechanism} $\mu :T\rightarrow \Delta (A)$ assigns to each type profile $%
t=(t_{1},t_{2})$ a probability distribution $\mu (t)$ over the set of
alternatives. We write $\mu ^{a}$ (resp. $\mu ^{0}$) for the probability
that alternative $a$ (resp. $a_{0}$) is chosen. For each player $i$, the
\textit{interim} expected probability that alternative $a\in A$ is chosen if
she has type $t_{i}$ but reports $\hat{t}_{i}$ while player $j$ reports
honestly is given by
\begin{equation}
\mu _{i}^{a}(\hat{t}_{i}|t_{i})=\sum_{t_{j}\in T_{j}}\mu ^{a}(\hat{t}%
_{i},t_{j})f_{i}^{c}(t_{j}|t_{i}),
\end{equation}%
and the corresponding interim expected utility is given by%
\begin{equation}
U_{i}^{\mu }(\hat{t}_{i}|t_{i})=\mu _{i}^{0}(\hat{t}_{i}|t_{i})t_{i}+\sum_{a%
\in A^{0}}\mu _{i}^{a}(\hat{t}_{i}|t_{i})u_{i}^{a}.
\end{equation}

Mechanism $\mu $ is \textit{incentive compatible} (IC) if reporting
truthfully is a Bayesian Nash equilibrium. If $\mu $ is incentive
compatible, we denote the \textit{ex-post} equilibrium utility of player $i$
by $u_{i}^{\mu }(t)$; the \textit{interim} probability of $a\in A$ given $%
t_{i}$ and truthtelling by $\mu _{i}^{a}(t_{i})$, and the \textit{interim}
utility from truthtelling by $U_{i}^{\mu }(t_{i})$. From now on, when we
refer to a mechanism, we always mean that it is physically feasible ($\mu
(t)\in \Delta (A)$ for all $t\in T$) and incentive compatible.

We note that, since the players already know their types at the start of the
interaction and each one can unilaterally enforce disagreement, only
mechanisms that are \textit{individually rational} (IR), i.e., that satisfy $%
U_{i}^{\mu }(t_{i})\geq t_{i}$ for all $t_{i}\in T_{i}$ and $i=1,2$, can
actually be implemented. Although IR can, therefore, be seen as a necessary\
requirement for feasibility, some of our results do not depend on this
assumption and we will mention this property only when we actually need it.

We will mainly be interested in \textit{efficient mechanisms}. An efficient
mechanism generates a Pareto efficient outcome with probability 1.
Efficiency is a very desirable property: if $\mu $ is inefficient, then
players know from the\ start that, if they accept $\mu$, with positive
probability, they will have an incentive to renegotiate; foreseeing this
they\ may reject this mechanism. This argument has special force if $f$ has
full support, since in that case, the players even know a range of outcomes
to which they could renegotiate.

\begin{df}
Mechanism $\mu $ is efficient if $\Pr_{f}\{t:u^{\mu }(t)\in \partial U\}=1$.
It is strongly efficient if $u^{\mu }(t)\in \partial U$ for all $t\in T$.
\end{df}

Our concept of strong efficiency coincides with the classical \textit{ex-post%
} efficiency concept from \cite{HM83}. That paper discusses six efficiency
concepts, which differ depending on whether or not incentive constraints are
taken into account (classical or incentive efficiency) and on the point in
time (\textit{ex-ante}, \textit{interim} or \textit{ex-post}) at which the
evaluation is taking place. The next Proposition links efficiency to \textit{%
ex-ante} incentive efficiency; together with Proposition 2, it gives the
reader a better idea of our model. The proof is given in the Appendix.%
\footnote{%
A mechanism $\mu $ is \textit{ex-ante }incentive efficient if there does not
exist another mechanism $\nu $ which, when types are not yet known, both
players consider at least as good as $\mu $, with at least one player
viewing it as strictly better. A mechanism $\mu $ is \textit{interim }%
incentive efficient if there does not exist another mechanism that all types
of both players (weakly) prefer to $\mu $, with the preference being strict
for at least one type. Our proof shows that Proposition 1(2) also holds if
\textit{ex-ante} is replaced by \textit{interim}.}

\begin{prop}
Let $\Gamma $ be a bargaining game with $f$ having full support.

(1) If $\partial U$ is linear, an efficient mechanism is ex-ante incentive
efficient.

(2) If $\partial U$ is non-linear, this need not hold.

(3) There always exists a mechanism that is individually rational, efficient
and ex-ante incentive efficient.
\end{prop}

In a bargaining game $\Gamma $ in which $f$ has full support, there exist
many individually rational and efficient mechanisms; any constant mechanism
that Pareto dominates $a_{0}$ satisfies these conditions. On the other hand,
if $f$ does not have full support, no such mechanism may exist.

\begin{prop}[Existence and multiplicity]
Let $\Gamma$ be a bargaining game.

(1) If $f$ has full support, there exist multiple individually rational and
strongly efficient mechanisms.

(2) If $f$ does not have full support, an individually rational and strongly
efficient mechanism may not exist.
\end{prop}

The proof of the second part of this Proposition is deferred to Section 3,
where we will show that it follows almost immediately from Theorem 1.
Because of Proposition 2(2), we will focus mainly on the case with full
support. Note that, when the prior $f$ is independent, our assumption that $%
f_{i}(t_{i})>0$ for all $t_{i}$ implies that $f$ has full support.

In this paper, we investigate whether there exist efficient mechanisms that
depend non-trivially on the players' types, i.e., that take into account
some of the private information that the players have. As the game theoretic
approach to bargaining is welfaristic \citep{SE79}, real outcomes only
matter through their utilities; hence, we focus on the \textit{interim}
utilities associated with a mechanism. We say that a mechanism is \textit{%
type-dependent} if for at least one of the players, the \textit{interim}
utility function is not constant.

\begin{df}
Mechanism $\mu $ is type-dependent if, for at least one player $i$, the
interim utility function $t_{i}\rightarrow U_{i}^{\mu }(t_{i})$ is not a
constant function; it is type-independent if both $U_{1}^{\mu }$ and $%
U_{2}^{\mu }$ are constant functions.
\end{df}

We note that type-independence is a much weaker property than the mechanism
prescribing the same lottery for each $t$. Type-independence relates to the
\textit{interim} stage and there can be much variation \textit{ex-post}
without this influencing a player's \textit{interim} utility.\footnote{%
An explicit example to illustrate this fact is available from the authors
upon request.}

It is easily seen that, if the players' disagreement payoffs are
statistically independent, an efficient mechanism must be type-independent.
Efficiency\ requires that $\mu ^{0}(t)=0$ for all $t$ and if $\mu $
prescribes a more attractive lottery on $\partial U$ for type $t_{i}$ than
for some other type $t_{i}^{\prime }$, then clearly it is not incentive
compatible. Hence, we have the following result, which is independent of the
shape of the Pareto frontier, also holds if there are more than two players,
and also holds when $T_{i}$ is an interval. This result provides a
benchmark; from now on we will assume that $f$ is not independent. {\ }

\begin{obs}
If $f$ is independent ($f(t)=f_{1}(t_{1})f_{2}(t_{2})$ for all $t$), then
any efficient mechanism is type-independent.\footnote{\cite{MI12} has
derived a somewhat similar result in a different context. He studies the
allocation of two (identical) objects to two (\textit{ex-ante} identical)
players when monetary transfers are impossible and each player must get an
object. He shows that when the players' values are independent, one cannot
do better than random allocation; hence, the intensity of the players'
preferences cannot play a role.}
\end{obs}

\section{ The linear case}

Theorem 1 below is one of the main results of the paper. It states that the
above observation continues to hold when beliefs are correlated, provided
that $f$ has full support and the Pareto frontier $\partial U$ is linear.
Hence, even when type-dependence appears natural and appealing, such as when
types are strongly negatively correlated, it is incompatible with
efficiency. Although the case of a linear Pareto frontier is somewhat
special, it is interesting and relevant. It arises, for example, when there
is pure conflict, i.e., there are just three alternatives, $%
A=\{a_{0},a_{1},a_{2}\}$ with $a_{i}$ defined as above. Linearity also
results when the players can use monetary transfers to reach a compromise, a
case which we will discuss in more detail in Section 5.

\begin{tm}
If the Pareto frontier is linear and $f$ has full support, then any
efficient mechanism is type-independent.\footnote{%
We conjecture that Theorem 1 continues to hold when each $T_{i}$ is an
interval in $%
\mathbb{R}
$ and $f$ admits a continuous density. As there is only a small literature
on correlated equilibria in non-finite games, following \cite{HS89}, and
there are some subtle issues, we did not pursue this direction.}
\end{tm}

\begin{proof}[Proof of Theorem 1]
Let $\Gamma $ be a bargaining game as defined in Section 2 and assume $%
\partial U$ is linear. Without loss of generality, assume $\Gamma $ is
normalized such that $u_{1}+u_{2}=1$ for each $u\in \partial U$. Let $\mu $
be an efficient mechanism. Given the probability distribution $f$ on $T$,
define $G^{\mu }=<T_{1},T_{2},u_{1}^{\mu },u_{2}^{\mu }>$ as the game in
which the set of pure strategies of player $i$ is $T_{i}$ ($i=1,2$) and in
which $i$ gets payoff $u_{i}^{\mu }(t)$ if $t\in T$ is played. It is easy to
see that $\mu $ is incentive compatible if and only if $f$ is a correlated
equilibrium of $G^{\mu }$. Note that, since $f$ has full support and $\mu $
is efficient, $G^{\mu }$ is a constant-sum game: $u_{1}^{\mu }(t)+u_{2}^{\mu
}(t)=1$ for all $t\in T$. Let $v_{i}\in R$ be the value of this game for
player $i$; hence, $v_{1}+v_{2}=1$. For general (finite) zero-sum games \cite%
{FO90} already observed that the beliefs associated with pure strategies
that are played with positive probability in a correlated equilibrium are
optimal strategies, i.e., they yield at least the value against any strategy
of the opponent. The next lines of our proof build on (or follow from) this
insight. Note that, since $f$ has full support, all strategies $t_{i}$ are
played with positive probability in the correlated equilibrium $f$.
Furthermore, as a player can always play a maximin strategy (rather than
follow a recommendation to play $t_{i}$), we must have $u_{i}^{\mu
}(t_{i},f_{i}^{c}(t_{i}))\geq v_{i}$ for each $i$ and each $t_{i}$. None of
these inequalities can be strict, as in that case we would have $%
u_{i}(f)>v_{i}$ for the \textit{ex-ante} expected payoff of player $i$ in $f$
and, hence, $u_{j}(f)<v_{j}$ for $j\neq i$, which is impossible since player
$j$ can guarantee himself $v_{j}$. Hence, we have $U_{i}^{\mu
}(t_{i})=u_{i}^{\mu }(t_{i},f_{i}^{c}(t_{i}))=v_{i}$ for all $i$ and $t_{i}$%
. Consequently, the mechanism $\mu $ is type-independent.
\end{proof}

Following \cite{CM88}, we say that $f$ satisfies the full rank condition if,
for $i=1,2$, the set $\{f_{i}^{c}(t_{i})|t_{i}\in T_{i}\}$ of posterior
beliefs of player $i$ is linearly independent. Note that this condition is
generically satisfied if $|T_{1}|=|T_{2}|$. The next result shows that, if
the beliefs $f$ satisfy this condition and have full support, any efficient
mechanism $\mu $ must yield constant \textit{ex-post} utility, i.e., $%
u_{i}^{\mu }(t)=u_{i}^{\mu }(t^{\prime })$ for all $t$, $t^{\prime }$ and $i$%
. The intuition behind this result is that, if each $t_{i}$ appears with
positive probability in the correlated equilibrium, all incentive
constraints must hold with equality. Writing $n=|T_{1}|=|T_{2}|$, this gives
$2n(n-1)$ equations in $u^{\mu }$. On top of this, we have the $2n$ interim
utility constraints, $U_{i}^{\mu }(t_{i})=v_{i}$, encountered in the proof
of Theorem 1, which together yields $2n^{2}$ equations, just as much as
there are variables in $u^{\mu }$. The CM-condition guarantees that the
equations are independent; hence, there is a unique solution, which must be
constant.{\ }

\begin{prop}
If the Pareto frontier is linear, $f$ has full support and the Cr\'{e}%
mer-McLean full rank condition holds, then any efficient mechanism $\mu $ is
constant, i.e., if $t,t^{\prime }\in T$, then {$u_{i}^{\mu }(t)=u_{i}^{\mu
}(t^{\prime })$} for $i=1,2$.
\end{prop}

In the remainder of this Section we make four remarks that serve to put
Theorem 1 in perspective.\vspace{2mm}

\noindent \textbf{Remark 1: }$\boldsymbol{T}$\textbf{\ must be a product set.%
}

\noindent The assumption that $T$ is a product set, $T=T_{1}\times T_{2}$ is
routinely made in the literature (see e.g., \citet{BO15}, p.174). Theorem 1,
however, does not hold if this assumption is not satisfied. For example,
efficient type-dependent solutions can exist if the disagreement points lie
on a curve $C=\{t|t_{2}=\varphi (t_{1})\}$, i.e., the case of perfect
correlation in which Spearman's rank correlation coefficient $\rho =\pm 1$.
A simple example is as follows. Assume that each player can be of 2 types, $%
l $ or $h$, with $l<h$ and $f(l,h)=f(h,l)=\frac{1}{2}$. Consider the
symmetric mechanism $\mu $ that prescribes the lottery $p\ast
a_{1}+(1-p)\ast a_{2}$ (with $p>\frac{1}{2})$ if $(h,l)$ is reported and
disagreement if $(l,l)$ or $(h,h)$ is reported. If $h<p<1-l$, then $\mu $ is
incentive compatible and efficient, with $U^{\mu }(l)<U^{\mu }(h)$; hence, $%
\mu $ is type-dependent. A similar example can be constructed if the
disagreement outcomes are known to be very bad, i.e., if for each $t$ with $%
f(t)>0$, we have $t_{1}+t_{2}<\frac{1}{2}$.\vspace{2mm}

The above example relies on the mechanism prescribing an inefficient outcome
if the players report a type combination that has prior probability $0$.
This implies that the game $G^{\mu }$, defined in the proof of Theorem 1, is
no longer constant-sum. If we insist on having efficient outcomes for all
type combinations, $G^{\mu }$ is again zero-sum and the proof remains valid;
hence, we have:

\begin{cor}
If the Pareto frontier is linear, any strongly efficient mechanism is
type-independent.
\end{cor}

This corollary allows us to prove the following Proposition, which in turn
proves the second part of Proposition 2 from Section 2.

\begin{prop}
If the Pareto frontier $\partial U$ is linear and there does not exist a
lottery on $\partial U$ that Pareto dominates all $t$ in the support of $f$,
then there does not exist a strongly efficient and individually rational
mechanism.
\end{prop}

\begin{proof}[Proof of Proposition 4]
Let the conditions be satisfied and assume that $\mu $ is a strongly
efficient mechanism. Then $\mu $ is type-independent, hence, it produces a
constant interim utility pair $u=(u_{1},u_{2})\in \partial U$. By the
conditions of the Proposition, there exist some $t\in T$ that is not Pareto
dominated by $u$. Hence, we must have $t_{1}\geq u_{1}$ or $t_{2}\geq u_{2}$%
. Since $t$ is not Pareto efficient, $t\neq u$; hence, one of these
inequalities must be strict. But if $t_{i}>u_{i}$, then $\mu $ is not
individually rational.
\end{proof}

\noindent \textbf{Remark 2: The result only holds for two players. }

\noindent Theorem 1 no longer holds if there are more than two players
(unless $f$ is independent). In the Appendix we construct a symmetric,
three-player example with correlated beliefs in which an efficient and
type-dependent mechanism exists.\vspace{2mm}

\noindent \textbf{Remark 3: Epsilon-efficiency}

\noindent With a weaker efficiency concept, type-dependent solutions become
possible. For example, define mechanism $\mu $ to be $\varepsilon $%
-efficient if $\Pr {}_{f}\{t|\mu ^{0}(t)\leq \varepsilon \}=1$; hence,
whatever the state of the world, disagreement occurs at most with
probability $\varepsilon $. It is easily seen that such $\varepsilon $%
-efficiency is compatible with type-dependence. We do not consider $%
\varepsilon $-efficiency to be an attractive concept. In the Introduction
and in Section 2 we argued that efficiency is a very desirable property.
That argument implies that an $\varepsilon $-efficient mechanism is just as
unsatisfactory as any other mechanism that is not fully efficient.\vspace{2mm%
}

{\noindent \textbf{Remark 4: Indirect mechanisms} }

\noindent Theorem 1 implies that any efficient incentive compatible outcome
can be implemented by a constant lottery, the most simple indirect
mechanism. Hence, players do not need to reveal any information about their
disagreement payoff. In contexts in which such information is sensitive,
this can be a desirable property.{\noindent }

\section{ The non-linear case}

In this Section, we investigate to what extent Theorem 1 can be generalized
if the Pareto frontier is non-linear (i.e., if it is only piece-wise
linear). We show that the result continues to hold if at least one player
has at most two types. We also provide an example in which each player has
three types and in which an efficient, individually rational, type-dependent
solution exists.

\begin{tm}
Suppose the Pareto frontier is non-linear and $f$ has full support. If $\min
(|T_{1}|,|T_{2}|)=2$, an efficient mechanism is type-independent. If $\min
(|T_{1}|,|T_{2}|)>2$, an efficient, individually rational and type-dependent
mechanism can exist.
\end{tm}

\begin{proof}[Proof of Theorem 2]
The proof of the first part of the theorem is given in the Appendix. To
prove the second part, we give a symmetric example with $|T_{1}|=|T_{2}|=3$.
Let $A=\{a_{0},a_{1},a_{2},a_{3}\}$ with $a_{0}$, $a_{1}$ and $a_{2}$ as in
Section 2, let the utility functions be normalized as in Section 2, and
assume $u_{i}(a_{3})=\frac{7}{10}$ for $i=1,2$. Furthermore, assume that
each player has three types $s_{1}$, $s_{2}$ and $w$, with $u_{i}^{0}(t)\leq
\frac{1}{2}$ for all $t$. The prior probability distribution over the
players' types is given in Table 1 where $\varepsilon >0$ is assumed to be
small.

\bigskip
\begin{tabular}{c|c|c|c}
$f$ & $s_{1}$ & $s_{2}$ & $w$ \\ \hline
$s_{1}$ & $\varepsilon $ & $\varepsilon $ & $\frac{1}{4}-2\varepsilon $ \\
\hline
$s_{2}$ & $\varepsilon $ & $\varepsilon $ & $\frac{1}{4}-2\varepsilon $ \\
\hline
$w$ & $\frac{1}{4}-2\varepsilon $ & $\frac{1}{4}-2\varepsilon $ & $%
4\varepsilon $%
\end{tabular}
\bigskip

Table 1 \bigskip

Note that the situation is fully symmetric and that the players' beliefs are
almost perfectly negatively correlated:
\begin{align*}
f_{i}^c(w|s_{1})=f_{i}^c(w|s_{2})\approx 1,\,\,\,f_{i}^c(w|w)\approx 0.
\end{align*}

Consider the following symmetric mechanism $\mu $. If each player reports a
type in $\{s_{1},s_{2}\}$ and the reports differ, then $a_{1}$ is selected;
on the other hand, $a_{2}$ is selected if the reports are the same. If
player 1 reports a type in $\{s_{1},s_{2}\}$ and player 2 reports type $w$,
then the lottery $(\frac{1}{6},\frac{5}{6})$ over $a_{1}$ and $a_{3}$ is
selected. Symmetry dictates the lottery $(\frac{1}{6},\frac{5}{6})$ over $%
a_{2}$ and $a_{3}$ if the roles of the players are reversed. Finally, if
each player reports $w$, then $a_{3}$ is selected for sure. Clearly, $\mu $
is efficient. Then the payoffs associated with this mechanism are as in
Table 2: \bigskip

\begin{tabular}{c|c|c|c}
$u^{\mu }$ & $s_{1}$ & $s_{2}$ & $w$ \\ \hline
$s_{1}$ & $(0,1)$ & $(1,0)$ & $(\frac{3}{4},\frac{7}{12})$ \\ \hline
$s_{2}$ & $(1,0)$ & $(0,1)$ & $(\frac{3}{4},\frac{7}{12})$ \\ \hline
$w$ & $(\frac{7}{12},\frac{3}{4})$ & $(\frac{7}{12},\frac{3}{4})$ & $(\frac{7%
}{10},\frac{7}{10})$%
\end{tabular}
\bigskip

Table 2\bigskip

Assume $\varepsilon $ to be small. If both players report honestly then $%
U_{i}^{\mu }(s_{1})=U_{i}^{\mu }(s_{2})\approx \frac{3}{4}$ and $U_{i}^{\mu
}(w)\approx \frac{7}{12}$; hence $\mu $ is type-dependent. Clearly, $\mu $
is also individually rational. If the other player tells the truth, then
reporting $s_{i}$ instead of $s_{j}$ results in the same payoff, while
misreporting as $w$ yields $\approx \frac{7}{10}<\frac{3}{4}$; if type $w$
reports $s_{i}$, then his expected payoff is $\approx \frac{1}{2}<\frac{7}{12%
}$; hence, $\mu $ is incentive compatible.
\end{proof}

Note that, in this example, we did not have to specify the disagreement
utilities of the types $s_{1},s_{2}$ and $w$; the example works as long as $%
t_{i}=u_{i}^{0}(t)\leq \frac{1}{2}$ for all $i$ and $t$. Hence, it is
allowed that $u_{i}^{0}(s_{1},t_{2})=u_{i}^{0}(s_{2},t_{2})$ for all $t_{2}$%
. In this case, $s_{1}$ and $s_{2}$ only differ in name: these types have
the same preferences and beliefs, and they can be thought of as a type $s$
having been split into two. One of the axioms \cite{HS72} use to
characterize their generalized Nash solution is \textquotedblleft player
splitting"; it insists that such splitting (keeping utilities and beliefs
the same) should not change the solution; see also \cite{WE92}. Our example
shows that \textquotedblleft player splitting" may not be as innocent as it
first might look: as soon as a player is split into two different decision
making agents, there can be coordination problems between them, and this can
change the outcome.

\section{Transferable utility}

In the model from Section 2, the players only have lotteries at their
disposal to reach a compromise. In this Section, we investigate the same
problem, but now with the assumption that compromises can also be reached by
making monetary transfers. Specifically, as in \cite{MS83} we assume that
the individuals are risk neutral and have additively separable utility for
money and the alternatives from $A$. Hence, we assume that each player $i$'s
utility function is given by $u_{i}(a,m_{i})=v_{i}(a)+m_{i}$, where $%
v_{i}(a) $ denotes the VNM-utility that player $i$ attaches to the
alternative $a\in A $ and $m_{i}$ is the amount of money that player $i$
holds \textit{ex-post}. In this specification, utility is transferable: if
player $i$ gives one unit of money to player $j$, $u_{i}$ decreases with one
unit while $u_{j}$ increases by $1$. Note that the assumption that money\
enters each player's utility function with coefficient $1$ implies that
utility becomes interpersonally comparable. More precisely, utility gains
(and losses) can be compared as each such difference is equivalent to a
certain amount of money, which both players value in the same way.

In the transferable utility literature, it is usually assumed that
constraints on money holdings are not binding\footnote{%
For example, \citet[p. 422]{MY91} defines transferable utility as there
being \textquotedblleft a commodity-called money-that players can \textit{%
freely} transfer among themselves, such that any player's utility payoff
increases by one unit for every unit of money that he gets." (Emphasis
added.) In cooperative game theory, TU is simply defined as, for each
coalition, there being a surplus that can be arbitrarily divided, while in
the matching literature it has been directly defined as the Pareto frontier
being a straight line with slope -1 (\citealp{CH17}, p.6, 27). The mechanism
design literature usually also assumes that transfers can be unbounded,
although it has been noted that this may influence the results (%
\citealp{CM88}, p. 1255).} which, among others, implies that Pareto
efficiency is equivalent to maximizing the sum of the players' utilities, $%
u_{1}+u_{2}$. In this case, the Pareto frontier is linear and Theorem 1
applies. In this Section, we investigate the more general case in which
budgets can be limited. This not only is more realistic, it is also
interesting since the Pareto frontier then can be non-linear. We will see
that a result as in Theorem 1 holds given some conditions on the players'
budgets, but certainly not in general.

\subsection{Model}

We first describe how the model from Section 2 is changed to fit the
TU-context. Let $A,a_{0},T$ and $f$ be as in Section 2. For each player $i$,
let $u_{i}$ be as above and assume that this player has an initial endowment
(budget) $b_{i}\geq 0$ of money. A \textit{bargaining problem with
transferable utility }is a tuple $\Gamma =(A,a_{0},T,f,v,b)$\textit{. }Note
that when player $i$ has budget $b_{i}$, the utility that his type $t_{i}$
assigns to disagreement is $t_{i}+b_{i}$. An (\textit{ex-post}) \textit{%
allocation} for $\Gamma $ is a triple $(a,m_{1},m_{2})$ with $a\in A$, $%
m_{i}\geq 0$ (for $i=1,2$) and $m_{1}+m_{2}=b_{1}+b_{2}$. As an allocation
is fully determined by the pair $(a,m_{1})$, we simplify notation by writing
$(a,m_{1})$ instead of $(a,m_{1},m_{2})$. We write $u(a,m_{1})$ for the
utility pair associated with the allocation $(a,m_{1})$; hence $%
u(a,m_{1})=(v_{1}(a)+m_{1},v_{2}(a)+b_{1}+b_{2}-m_{1})$. We say an
allocation is \textit{utilitarian} if it maximizes the sum of the players'
utilities, while $a\in A$ is called utilitarian if $(a,b_{1})$ is
utilitarian.

The players' ideal points $a_{1}$ and $a_{2}$ are defined as in Section 2,
and we continue to assume that $\underline{t}_{i}>v_{i}(a_{j})$ for each
player $i$ and $j\neq i$ (so that player $i$ will accept $a_{j}$ only when
he gets sufficient financial compensation). Given that utility now is
interpersonally comparable, we can no longer normalize $v_{1}$ and $v_{2}$
as in Section 2, as this would imply that $a_{1}$ and $a_{2}$ are equally
good from the social point of view. However, without loss of generality, we
can still normalize $v_{1}$ and $v_{2}$ such that $%
v_{1}(a_{2})=v_{2}(a_{1})=0$ and $v_{1}(a_{1})=1$. The former just pins down
the level of zero utility while the latter simply amounts to choosing a
convenient unit of account. With respect to disagreement, we assume that
there exists $a^{\ast }\in A$ such that
\begin{equation*}
t_{1}+t_{2}<v_{1}(a^{\ast })+v_{2}(a^{\ast })\,\,\text{for all}\,\,t\in T.
\end{equation*}

Note that this assumption is slightly weaker than the one we made in Section
2, as this inequality still allows $t=(t_{1},t_{2})$ to be Pareto efficient
when transfers are not possible. Clearly, all allocations with $a=a^{\ast }$
are Pareto efficient. On top of this, there can be other efficient
allocations, which are not utilitarian. For example, if $%
v_{1}(a_{2})+v_{2}(a_{2})<v_{1}(a_{1})+v_{2}(a_{1})$, the allocation $%
(a_{2},0)$ still is Pareto efficient: it is the best possible outcome for
player $2$ and player $1$ does not have any money to induce player $2$ to
accept $a_{1}$. In this case, if $A=\{a_{0},a_{1},a_{2}\}$ and $a_{1}$ is
utilitarian, all lotteries on the set $\{(a_{1},0),(a_{2},0)\}$ are Pareto
efficient. Below, we write $\partial U(b)$ for the Pareto frontier generated
by $\Gamma $, $U(b)$ for the utilitarian part of the frontier and $NU(b)$
for the non-utilitarian part.

To cover bilateral trade as in \cite{MS83}, we also allow {$|A|=2$; hence} $%
A=\{a_{0},a^{\ast }\}$, with $a^{\ast }$ as in the inequality above. In this
case, without loss of generality, we can assume that player $1$ prefers $%
a^{\ast }$ to $a_{0}$ , which implies that $a^{\ast }=$ $a_{1}$. To see that
this case indeed covers bilateral trade, view player $1$ as the buyer and
player $2$ as the seller, and interpret $a_{0}$ and $a^{\ast }$ as no trade
and trade, respectively. Then t{he only difference with \cite{MS83} is that
we assume it is common knowledge that trade is necessary to maximize
surplus. Note that, in this case, the allocation }$(a_{0},0)$ in which the
seller keeps the object and gets all the money, although not utilitarian, is
still Pareto efficient: $U(b)$ is the line segment between $u(a^{\ast },0)$
and $u(a^{\ast },b_{1}+b_{2})$, while $NU(b)$ is the line segment between $%
u(a_{0},0)$ and $u(a^{\ast },0)$.

In the remainder of this Section, we distinguish two cases: (1) At least one
of the $a_{i}$ is utilitarian (which, without loss of generality we assume
to be $a_{1}$); (2) None of the $a_{i}$ is utilitarian. We will see that, in
the first case, a result as in Theorem 1 continues to hold provided that
player $1$ has high enough budget as compared to player $2$, and that the
second case is very much like the non-linear case discussed in Section 4.

\subsection{At least one ideal point is utilitarian}

Let $\Gamma =(A,a_{0},T,f,v,b)$\textit{\ }be a bargaining problem with
transferable utility and, without loss of generality, assume that $a_{1}$ is
utilitarian.{\ If }$a_{2}$ is also utilitarian, then the Pareto frontier is
linear and Theorem 1 applies. We will refer to this case as a \textit{%
surplus division problem}. There are two other cases: $|A|\geq 3$ and $a_{2}$
is not utilitarian, or $|A|=2$ and {$A=\{a_{0},a_{1}\}$} (bilateral trade).
In both of these cases, Theorem 1 cannot be directly applied as the Pareto
frontier is non-linear. Nevertheless, we will show that, if the budget of
player $1$ is sufficiently large relative to that of player $2$, an
efficient and individually rational mechanism must exclusively select
utilitarian allocations. This implies that the proof of Theorem 1 can be
applied and that the result continues to hold. Hence, {we have the following
theorem.}

\begin{tm}
Let $\Gamma $ be a \textit{bargaining problem with transferable utility and
a prior }$f$ with full support.

(1) If both $a_{1}$ \textit{and }$a_{2}$ are utilitarian (hence, $\Gamma $
is a surplus division problem), then any efficient mechanism is
type-independent.

(2) If only $a_{1}$ is utilitarian, then, if $b_{1}>\frac{v_{1}(a_{1})+b_{2}%
}{\min_{t}f_{1}^{c}(t_{2}|t_{1})}-b_{2}$, any individually rational and
efficient mechanism is type-independent.
\end{tm}

\begin{proof}[Proof of Theorem 3]
It suffices to prove the second statement. First assume $|A|\geq 3$ and $%
a_{2}$ is not utilitarian. Let $a^{m}$ be the alternative from $A$ that is
the least preferred by player $1$ from the set of utilitarian alternatives.
The utilitarian part of the frontier, $U(b)$, then is the line segment
between $u(a^{m},0)$ and $u(a_{1},b_{1}+b_{2})$. Furthermore, it is easily
seen that if $a$ is not utilitarian, $(a,m_{1})$ can only be Pareto
efficient if $m_{1}=0$. Let $\mu $ be an efficient mechanism and suppose
that there exists some $t=(t_{1},t_{2})$ for which $\mu $ selects a point
from $NU(b)$ with positive probability. For this specific $t_{1}$, write $%
\Pi $ for the set of $t_{2}$ for which $\mu $ selects a non-utilitarian
lottery. Then $\Pi \neq \emptyset $ and for all such lotteries we have that
the expected payoff of player $1$ is less than $u_{1}(a^{m},0)=v_{1}(a^{m})$%
. Simplifying notation by writing $p(t_{2})=f_{1}^{c}(t_{2}|t_{1})$ and $\pi
=\sum_{t_{2}\in \Pi }p(t_{2})$, we have:
\begin{align*}
U_{1}(t_{1})<& \sum_{t_{2}\in \Pi }p(t_{2})v_{1}(a^{m})+\sum_{t_{2}\in
T_{2}\backslash \Pi }p(t_{2})(v_{1}(a_{1})+b_{1}+b_{2}) \\
=& \,\,\pi v_{1}(a^{m})+(1-\pi )(v_{1}(a_{1})+b_{1}+b_{2}) \\
\leq & \,\,v_{1}(a_{1})+(1-\pi )(b_{1}+b_{2}).
\end{align*}%
For $\mu $ to be individually rational, we must have $U_{1}^{\mu
}(t_{1})\geq t_{1}+b_{1}\geq b_{1}$; hence
\begin{eqnarray*}
v_{1}(a_{1})+(1-\pi )(b_{1}+b_{2}) &\geq &b_{1}\text{, or} \\
v_{1}(a_{1})+b_{2} &\geq &\pi (b_{1}+b_{2}).
\end{eqnarray*}%
Since $\pi \geq \min_{t_{2}}f_{1}^{c}(t_{2}|t_{1})\geq
\min_{t}f_{1}^{c}(t_{2}|t_{1})$, this condition can only be satisfied if
\begin{align*}
\min_{t}f_{1}^{c}(t_{2}|t_{1})(b_{1}+b_{2})\leq & \,\,v_{1}(a_{1})+b_{2}%
\text{ or} \\
b_{1}\leq & \,\,\frac{v_{1}(a_{1})+b_{2}}{\min_{t}f_{1}^{c}(t_{2}|t_{1})}%
-b_{2}.
\end{align*}

But the condition from the Theorem states that this inequality does not
hold. This implies that, under the conditions of the Theorem, a mechanism
that is efficient and individually rational must select utilitarian
allocations with probability 1. But then the assumption that $f$ has full
support together with the argument from the proof of Theorem 1 implies that $%
\mu $ must be type-independent.

Next, assume $|A|=2$ and {$A=\{a_{0},a_{1}\}$. Then $a_{1}$ is the unique
utilitarian alternative, hence, $a^{m}=a_{1}$ and the proof proceeds exactly
as in the case } $|A|\geq 3${.}
\end{proof}

As mentioned, case (2) of this theorem covers situations of bilateral trade
in which it is common knowledge that maximizing surplus requires trade. For
this case, the theorem says that, if the buyer's budget is sufficiently
large compared to the seller's, efficiency and individual rationality
require trade with probability 1, and a constant payment from the buyer to
the seller. It is interesting that the lower bound on the buyer's budget
depends on her beliefs: if she attaches small probability to some seller
types, this bound may be quite large.

The \ reader may be surprised about the large bound on $b_{1}$, and also
that $b_{1}$ must be large relative to $b_{2}$. Isn't it possible to specify
a simple bound for $b_{1}$ that is independent of $b_{2}$? For example, in a
normalized problem (with $v_{i}(a)\in \lbrack 0,1]$ for all $i$ and $a$)
isn't it sufficient that $b_{1}\geq 1$? The answer is ``no" as the following
example illustrates. This example modifies the one {with almost perfectly
negatively correlated beliefs constructed in Theorem 2 to fit the current
context. The mechanism $\mu $ selects the non-linear part of the frontier
(specifically $(a_{2},0)$) only }when both players {are of type $s$ and
miscoordinate. As this happens with small probability, the selected outcome
usually is on the linear part of the frontier, and by choosing the transfer
payments appropriately, we can make sure that all incentive constraints and
all individual rationality constraints are satisfied. }

\begin{ex}
If $b_{2}=0$ and $b_{1}=1$, there exists a \textit{bargaining problem with
transferable utility }$\Gamma $ for which there exists an individually
rational, efficient and type-dependent mechanism.

To construct an example, we add monetary transfers to the game used in the
proof of Theorem 2. Specifically, let $\Gamma $ be given by: $%
A=\{a_{0},a_{1},a_{2}\}$ with $v_{1}(a_{1})=1$, $v_{2}(a_{2})=V<1$, $%
v_{1}(a_{2})=v_{2}(a_{1})=0;T_{i}=\{s_{1},$ $s_{2},$ $w\}$ and $f$ as in
Table 1, where $\varepsilon >0$ is assumed to be small; $b_{2}=0$ and $%
b_{1}=1$. We claim that, for appropriate values of the transfers $x$ and $y$
and disagreement utilities $u_{i}(s_{1})=u_{i}(s_{2})=s>u_{i}(w)=w$, the
mechanism $\mu $ described in Table 3 is individually rational, efficient
and type-dependent.

\begin{tabular}{c|c|c|c}
$\mu $ & $s_{1}$ & $s_{2}$ & $w$ \\ \hline
$s_{1}$ & $(a_{2},0)$ & $(a_{1},1)$ & $(a_{1},x)$ \\ \hline
$s_{2}$ & $(a_{1},1)$ & $(a_{2},0)$ & $(a_{1},x)$ \\ \hline
$w$ & $(a_{1},y)$ & $(a_{1},y)$ & $(a_{1},\frac{x+y}{2})$%
\end{tabular}%
\bigskip

Table 3\bigskip

It is easy to see that $\mu $ is efficient and type-dependent; hence, we
only have to check the incentive constraints and the individual rationality
constraints. Note that for $\varepsilon $ small, we have $U_{1}^{\mu
}(s)\approx 1+x,U_{1}^{\mu }(w)\approx 1+y,U_{2}^{\mu }(s)\approx 1-y$ and $%
U_{2}^{\mu }(w)\approx 1-x$. Hence, the IR constraints are satisfied for
player $1$ if $x,y>s$. They are satisfied for player $2$ if $x,y<1-s$. These
conditions hold for intermediate values of $x$ and $y$. If type $s_{k}$ of
player $1$ reports $w$, he expects $\approx 1+y$; hence, he does not deviate
if $x>y$. His type $w$ will not deviate if $y>0$. Finally, a type $s_{k}$ of
player $2$ will not deviate if $y<x$,\ while type $w$ of this player will
not deviate if $1-x>\frac{V+1}{2}$. It is easy to find values of the
parameters for which the system has a solution. For example, if $V=0.4$, $%
s=0.2$ and $w=0.1$, then any $x,y$ with $0.2<x<0.3$ and $0.2<y<x$ satisfies
all inequalities.
\end{ex}

\subsection{None of the ideal points is utilitarian}

We now consider the case in which neither $a_{1}$ or $a_{2}$ is utilitarian.
Let $a^{\ast }$ be a utilitarian alternative; hence, $a^{\ast }\neq
a_{1},a_{2}$. In this case, the Pareto frontier contains at least three
linear segments, one with slope $-1$, another with slope $<-1$ and one with
slope $>-1$. The condition from Theorem 3 already suggests that this case is
qualitatively different from the one in the previous subsection. On the one
hand, we must have $b_{1}>b_{2}$, to ensure that non-utilitarian Pareto
allocations between $a_{2}$ and $a^{\ast }$ are not selected, but a
symmetric argument requires $b_{2}>b_{1}$. Clearly, these conditions cannot
be satisfied simultaneously. {We can use the example from the proof of
Theorem 2 (Section 4) to construct an efficient, individually rational and
type-dependent mechanism: for arbitrary $(b_{1},$}$b_{2})${, change the
mechanism $\mu $ from that example by never letting players exchange any
money. Formally, consider the mechanism $\nu $ defined by $\nu _{1}(t)=(\mu
_{1}(t),b_{1})$ and $\nu _{2}(t)=(\mu _{2}(t),b_{2})$; then $\nu $ trivially
satisfies the three desired properties. }

\begin{cor}
Let $\Gamma $ be a bargaining problem with transferable utility in which
neither $a_{1}$ nor $a_{2}$ is utilitarian. For every $b_{1},b_{2}\geq 0$
and $f$ with full support, there can exist a mechanism that is efficient,
individually rational and type-dependent.
\end{cor}

\section{Applications to Bargaining Theory}

In this Section, we discuss three implications of our results for the theory
of bargaining under incomplete information. In our first application, we
consider non-transferable utility and two players with equal bargaining
power. In the second application, utility is transferable and, as in the
mechanism design literature, one of the players (the informed principal) has
all the bargaining power so that she can dictate the mechanism. We show that
in both cases, \textit{ex-post} efficiency results in a very small class of
implementable bargaining solutions. In the third application, rather than
focusing on efficiency, we put emphasis on fairness, and on proportional
solutions (for the NTU-case) and the egalitarian solution (in the TU-case).
Among others, we show that, if the beliefs are independent, an \textit{%
interim} egalitarian solution almost always selects disagreement. This is in
sharp contrast to the case with complete information in which there always
is an outcome that is efficient (all surplus is divided) as well as
egalitarian (the players get an equal share of the surplus).

\subsection{\noindent Symmetric players}

Thus far, we focused on a fixed bargaining problem: all elements of $\Gamma $
were assumed given. The axiomatic approach to bargaining theory derives its
strength from insisting that the solutions of different problems should be
related to (or consistent with) each other. In such axiomatizations, the
symmetry axiom plays an important role. For example, in Nash's (1950)
axiomatization of his complete information solution, the Independence of
Irrelevant Alternatives Axiom allows one to transform any problem in an
equivalent linear and symmetric one. The symmetry axiom, which merely states
that it does not matter which player is called player $1$, is undisputed and
can be easily extended to games with incomplete information. We introduce
the following notion of symmetry.

\begin{df}
A bargaining problem $\Gamma $ (as in Section 2) is symmetric if $%
\{u(a)|a\in A^{0}\}$ is symmetric, $T_{1}=$ $T_{2}$ and $f$ is symmetric, $%
f(t_{1},t_{2})=f(t_{2},t_{1})$ for all $t_{1},t_{2}\in T_{1}$. If $\Gamma $
is symmetric, then a mechanism $\mu $ for $\Gamma $ is symmetric if $%
u_{1}^{\mu }(t_{1},t_{2})=u_{2}^{\mu }(t_{2},t_{1})$ for all $t_{1},t_{2}\in
T_{1}$.
\end{df}

The symmetry axiom requires that the solution $\mu (\Gamma )$ of a symmetric
bargaining problem $\Gamma $ is symmetric (see for example, \citealp{HS72}; %
\citealp{WE92}).

Suppose that two players have equal bargaining power and that $\Gamma $ is a
symmetric bargaining problem with $f$ having full support. Theorem 1 implies
that, if $\partial U$ is linear, symmetry and efficiency lead to an
(essentially) unique solution. For example, if $A=\{a_{0},a_{1},a_{2}\}$,
the solution must be $\mu (t)=\frac{1}{2}a_{1}+\frac{1}{2}a_{2}$ for all $t$%
. (Note that this symmetric solution is individually rational). Since
symmetric problems play a crucial role in axiomatizations, we believe that
Corollary \ref{Co:3} will also prove to be relevant for asymmetric problems
with incomplete information.

\begin{cor}
\label{Co:3} Let $\Gamma $ be a symmetric, linear bargaining problem with $f$
having full support. If $\mu $ is a symmetric and efficient mechanism, then $%
\mu $ is interim equivalent to a 50/50 lottery over the players' most
preferred alternatives.
\end{cor}

\subsection{\noindent Monopoly and surplus extraction}

{Let $\Gamma $ be a surplus division problem with incomplete information ($%
|T_{1}|>1$ or $|T_{2}|>1).$} Suppose player $1$ has all the bargaining power
and can dictate the mechanism; hence, she is an \textit{informed principal}.
As in \cite{CM88} and \cite{SE08}, say that player $1$ can \textit{extract
all surplus} if there exists a mechanism that is efficient and that always
gives player $2$ his disagreement payoff.

It has been shown that, if types are independent, an informed principal
cannot benefit from her private information at all; hence, in this
situation, she cannot extract full surplus \citep{MT14}. In contrast, \cite%
{SE08} has shown that the informed principal can generically extract all
surplus if there are at least three players and the players' types are
correlated. Specifically, full extraction of the surplus is possible if two
conditions are satisfied: the convex independence condition (CI) and the
identifiability condition (KS).\footnote{%
Condition CI (convex independence) was introduced in \cite{CM88}. It
requires that, for $i=1,2$, there does not exist $t_{i}^{\prime }$ such that
$f_{i}^{c}(t_{i}^{\prime })$ belongs to the convex hull of the set $%
\{f_{i}^{c}(t_{i}):t_{i}\in T_{i}\setminus \{t_{i}^{\prime }\}\},$where $%
f_{i}^{c}(t_{i})$ denotes the conditional distribution of $f$ on $T_{j}$
given $t_{i}$.
\par
The identifiability condition (KS) was introduced in \cite{KS08}. A
distribution $f$ satisfies KS if, for all full support distributions $g\neq
f $, there is at least one $i$ and one $t_{i}^{\prime }\in T_{i}$ such that $%
g^{c}_i(t_{i}^{\prime })$ does not belong to the convex hull of the set $%
\{f_{i}^{c}(t_{i}):t_{i}\in T_{i}\setminus \{t_{i}^{\prime }\}\}$. \cite%
{SE08} focus on the case with more than $2$ players. In the $2$-person case,
one of their requirements is equivalent to $f$ having full support.} When
there are only two players, these conditions cannot be satisfied
simultaneously: CI cannot hold if the types are independent, while KS can
only hold in that case (\citealp{KS08}, p. 135). The literature has left
open the question of what an informed principal can achieve in the $2$%
-person case. Corollary \ref{Co:4} fills this gap, it shows that,
generically, not all surplus can be extracted, not even if types are
strongly correlated. In other words, the result of \cite{SE08} is sharp.

\begin{cor}
\label{Co:4} Let $\Gamma $ be a surplus division problem in which player $1$
has all the bargaining power and $|T_{2}|>1$. If $f$ has full support, then
player $1$ cannot extract all surplus.
\end{cor}

The result is easy to see. If $\mu $ is a mechanism that allows player $1$
to extract all surplus, then player $2$'s interim utility must be constant
because of Theorem 1, while on the other hand it should be given by $%
U_{2}^{\mu }(t_{2})=${$t_{2}+b_{2}$}. If $|T_{2}|>1$, these conditions
cannot be satisfied simultaneously.

\subsection{\noindent Egalitarianism implies inefficiency}

Throughout this paper, we focused on efficiency and did not discuss the
second important dimension of bargaining: the requirement that the outcome
be fair. Of course, fairness can be defined in various ways; however, when
there is transferable utility, equal division of the gains from cooperation
is unambiguous and natural. In the case of complete information, such equal
division ($u_{1}-t_{1}=u_{2}-t_{2}$) of the maximal surplus ($%
u_{1}+u_{2}-t_{1}-t_{2})$ is known as the egalitarian solution. If
information is complete, there always is a unique egalitarian outcome. When
the disagreement outcome is private information, {it is natural to look at
equality at the \textit{interim} stage. If the players' budgets are $b_{1}$
and $b_{2}$, t{his amounts to the requirement $U_{1}^{\mu
}(t_{1})-t_{1}-b_{1}={U_{2}^{\mu }}(t_{2})-t_{2}-$}$b_{2}$ for all $t_{1}$
and $t_{2}$. This equality is trivially satisfied by the mechanism that
always selects $a_{0}$, which also satisfies IR. However, since
egalitarianism requires type-dependence (in fact, type-monotonicity),
Theorem 1 immediately implies that egalitarian mechanisms must be
inefficient: in surplus division problems with incomplete information and $f$
having full support, efficiency requires {$U_{i}^{\mu }(t_{i})$ to be
independent of }}$t_{i}$. {When the prior $f$ is independent, egalitarianism
leads to a very inefficient outcome, as the next Proposition shows: The
outcome will be disagreement, except possibly in the lowest state }$%
\underline{t}=(\underline{t}_{1},\underline{t}_{2})$; when there is no
budget to redistribute ($b_{1}=b_{2}=0$), egalitarianism even leads to
disagreement with probability $1$.

\begin{prop}
\label{Pr:5} Let $\Gamma $ be a surplus division problem with budgets $b_{1}$
and $b_{2}$ and incomplete information ($|T_{1}|>1$ or $|T_{2}|>1).$

(1) If the prior $f$ has full support, there does not exist an efficient and
egalitarian mechanism.

(2) If $f$ is independent and $b_{1},b_{2}\geq 0$, then an egalitarian
mechanism $\mu $ selects the disagreement outcome for all $t\neq \underline{t%
}$.

(3) If in (2) $b_{1}=b_{2}=0$, then $U_{i}^{\mu }(t_{i})=t_{i}$ for all $%
t_{i}\in T_{i}$ and $i=1,2$.\footnote{\cite{CL10} has already shown that,
for social choice problems with incomplete information and with transfers,
if types are independent, then interim incentive efficiency and monotonicity
(as in \cite{KA77}) are incompatible. We obtain a similar result, using a
different monotonicity concept.}\vspace{2mm}
\end{prop}

{\noindent \textbf{Remark 5: The value of money.} }We note that in case (2)
of the Proposition, although $\mu^{0}(t)=1$ for all $t\neq \underline{t}
$ (hence, there is disagreement, unless each player happens to be of the
lowest type), we can still have $U_{i}^{\mu }(t_{i})>t_{i}+b_{i}$ for all $%
t_{i}$ and $i=1,2$, provided that $b_{1},b_{2}>0$. The reason is that, in an
egalitarian mechanism $\mu $, we can have $\mu^{0}(\underline{t})<1$,
so that the lowest types, when they meet, can possibly generate a surplus.
Of course, if only these types would get more than their disagreement
payoff, $\mu $ would not be egalitarian; however, if each player has some
money, then this can be used to share this surplus with the other types, so
as to get an egalitarian outcome. If $b_{1}=b_{2}=0$, such redistribution is
not possible, and for this reason we get a stronger result in case (3). The
following is an explicit (symmetric) example. Let each player have two
types, $w=0.1$ and $s=0.4$ which are equally likely. Assume that $A$
contains an agreement $\widehat{a}$ that gives each player utility $0.3$ and
that each player has a budget $b\geq 0.1$. Let $\mu $ be given by $\mu (w,w)=%
\widehat{a}$ and disagreement otherwise; furthermore, there are only
transfers when the players have different types, in which case $w$ gives $0.1
$ to $s$. This mechanism is egalitarian, satisfies IR and IC, and yields
utilities $U_{i}^{\mu }(w)=0.15+b$ and $U_{i}^{\mu }(s)=0.45+b$.

The egalitarian solution insists that the two players gain equally from
cooperation, hence, it is applicable only in situations in which the
players' utility levels can be measured on a common scale, i.e., if utility
differences are interpersonally comparable. This implies that the
egalitarian solution cannot be applied in the model of Section 2 as there we
assumed that the two scales can be varied independently. In that context,
proportional solutions are still meaningful, see \citep{KA77,MY77}. For the
model of Section 2, we define a mechanism $\mu $ to be (\textit{interim})
\textit{proportional}, if there exists a pair of weights $\lambda
_{1},\lambda _{2}>0$ such that $\lambda _{1}(U_{1}^{\mu
}(t_{1})-t_{1})=\lambda _{2}(U_{2}^{\mu }(t_{2})-t_{2})$ for all $t_{1}$ and
$t_{2}$. If there is incomplete information, i.e., $|T_{1}|>1$ or $|T_{2}|>1$, the Pareto frontier is linear and $f$ has full support, Theorem 1 implies
that, if $\mu $ is efficient, $U_{i}^{\mu }(t_{i})$ does not depend on $%
t_{i} $. Hence, we have a similar impossibility result as in Proposition
5(1).

\begin{cor}
\label{Co:5} If $\Gamma $ is a bargaining game with $|T_{1}|>1$ or $%
|T_{2}|>1 $, a linear Pareto frontier $\partial U$ and a prior $f$ with full
support, there does not exist an efficient and proportional mechanism.
\end{cor}

\section{Conclusion}

We made a minimal change to the canonical two-person (Nash) bargaining model
by allowing the players' disagreement payoffs to be private information,
while maintaining the assumption that it is common knowledge that
disagreement is inefficient. In the resulting model, multiple \textit{ex-post%
} efficient (and individually rational) mechanisms exist, provided that the
players' prior over the possible disagreement points (type-pairs) has full
support. In Nash's complete information model, all standard solution
concepts insist on efficiency and satisfy Disagreement Point Monotonicity
(DPM): if one player's disagreement utility increases, then this player's
bargaining outcome (weakly) improves. We investigated whether, if the
disagreement payoffs are private information, efficiency is compatible with
a weaker version of DPM, type-dependence, and showed that, when the Pareto
frontier is linear or the players' disagreement payoffs are independent,
these properties cannot be satisfied simultaneously. This holds both in a
non-transferable utility context as well as in a TU-context. When the
Pareto frontier is non-linear, however, type-dependent solutions become
possible.

We attach more weight to our impossibility results for the linear case than
to the possibility results of the non-linear one. The reason is that linear
problems play an important role in axiomatizations of bargaining solutions.
For example, in the complete information case, Nash's IIA axiom allows to
transform a general problem into a linear one. It seems likely that linear
problems will also play an important role in axiomatizations when there is
incomplete information.

The incompatibility of the two properties forces one to make a choice
between efficiency and type-dependence: which of these principles should an
impartial mediator use? The answer can depend on the context and on which
commitment possibilities the players (and the mediator) have. In our view,
when renegotiation is possible, \textit{ex-post} efficiency is indispensable
as an axiom: rational players will not agree to a mechanism of which they
know \textit{ex-ante} that they might renegotiate away from its
recommendation once it has been received. Hence, the bargaining solutions
proposed by Harsanyi-Selten and Myerson, which generally select inefficient
outcomes, do not seem viable in this context. On the other hand, as argued
in \cite{HM83} and \cite{MY84}, the \textit{ex-ante} utilitarian solution is
not appealing either since it does not adequately deal with intra-player
fairness issues.\footnote{%
Also see \cite{KI17}. This paper uses Myerson's neutral solution to study
the question which type of mediator two bargainers will choose in a
situation of incomplete information. The model is symmetric, with each
player being either $s$ or $w$, and these types having different
preferences. The paper shows that, to avoid information leakage, each type
will propose the mediator that is preferred by type $s$, but this mediator
is \textit{ex-ante} inefficient.} Since these are, essentially, the only
general solutions that have been proposed in the literature on bargaining
games with incomplete information, we can conclude that, at this stage, we
do not yet have a good solution concept for bargaining problems with
incomplete information in which the players have limited commitment
possibilities.

In the applied literature (for example, \citealp{FU81}) it has been stressed
that a bargainer foremost should aim to improve his disagreement outcome as
this yields a better final outcome. This is a direct application of
Disagreement Point Monotonicity. Although, in our model, the set of
disagreement payoffs is exogenous and not the subject of choice and, hence,
the model cannot directly be used to evaluate this recommendation when there
is incomplete information, our paper suggests that this recommendation is
incomplete: when player $i$'s disagreement payoff is private information,
the bargaining outcome may be independent of how good or how bad it is.

\newpage

\begin{center}
{\large \textbf{Appendix: Proofs not given in the main text}}
\end{center}

\begin{proof}[Proof of Proposition 1]
Consider a normalized bargaining problem $\Gamma $. (1) If $\partial U$ is
linear and $\mu $ is an efficient mechanism, then, since $f$ has full
support, $u_{1}^{\mu }(t)+u_{2}^{\mu }(t)=1$ for all $t\in T$. As there
cannot exist a mechanism with a total \textit{ex-ante} sum of payoffs larger
than $1$, $\mu $ is \textit{ex-ante} incentive-efficient.

(2) Assume $\Gamma $ has as set of alternatives $A=\{a_{0},a_{1},a_{2},a_{3}%
\}$ with $u_{i}(a_{3})=0.7$, for $i=1,2$. Assume that each player has two
types, $s$ and $w$ that are equally likely and independent. Let $\mu $ be
the mechanism that selects $a_{1}$ when the players report the same type and
$a_{2}$ when the players report different types. Then $\mu $ is efficient,
with \textit{ex-ante} utility 0.5 for each player. Hence, $\mu $ is \textit{%
ex-ante} dominated by the mechanism $\nu $ that always selects $a_{3}$.
(Note that $\nu $ also dominates $\mu $ at the \textit{interim} stage.)

(3) Our assumptions imply that the intersection of $\partial U$ and the
individually rational region for $\bar{t}$, denoted by $\partial U^{\ast }$,
is non-empty. Pick any point $u\in \partial U^{\ast }$. Since $\partial U$
is the boundary of the convex set $U$, there exists a supporting hyperplane
at $u$; hence, there exists $\lambda =(\lambda _{1},\lambda _{2})\in
R_{+}^{2}$ such that $\lambda \cdot u\geq \lambda \cdot v$ for all $v\in U$.
Let $L\in \Delta (A^{0})$ be a lottery that generates the payoff $u$, and
let $\mu $ be the constant mechanism given by $\mu (t)=L$ for all $t$. Then $%
\mu $ is efficient and individually rational. Furthermore, if $\nu $ is a
mechanism with \textit{ex-ante} payoff $U^{\nu }$, then $\lambda \cdot u\geq
\lambda \cdot U^{\nu }$; hence, $\mu $ is \textit{ex-ante} efficient.
\end{proof}

\begin{proof}[Proof of Proposition 2]
(1) Let $f$ have full support and write $\bar{t}=(\bar{t}_{1},\bar{t}_{2})$.
The situation is as in Figure 1(c). The lottery $\alpha (\bar{t})$ on $A^{0}$
Pareto dominates $t$ for all $t\in T$ and the same holds for all lotteries
that are close to $\alpha (\bar{t})$, and for all efficient lotteries that
Pareto dominate these (if they exist). Any mechanism $\mu $ that always
selects the same efficient lottery from this (infinite) set satisfies the
conditions from the Proposition.

(2) As Figure 1(b) shows, if $f$ does not have full support, there may not
be a lottery that Pareto dominates all $t\in T$. Proposition 4 in Section 3
shows that in such a case an individually rational and strongly efficient
mechanism may not exist.
\end{proof}

\begin{proof}[Proof of Observation 1]
Suppose $\mu $ is efficient. For a given strategy of player $j$ and each
message $t_{i}$ of player $i$, $\mu $ induces a lottery $\mu (t_{i})\in
\Delta (A^0)$. All types of player $i$ have the same preference relation
over this set of lotteries, hence IC implies that $U_{i}^{\mu }(t_{i})$ must
be independent of $t_{i}$. Hence, $\mu $ is type-independent.
\end{proof}

\begin{proof}[Proof of Proposition 3]
The proof of Theorem 1 established that $u_{i}^{\mu
}(t_{i},f_{i}^{c}(t_{i}))=v_{i}$ for all $i$ and $t_{i}$. In fact, since
each $t_{i}$ appears with positive probability (since $f$ has full support)
in the correlated equilibrium, all such incentive inequalities must hold
with equality. If $v_{i}>u_{i}^{\mu }(t_{i}^{\prime },f_{i}^{c}(t_{i}))$ for
at least one $t_{i}^{\prime }\neq t_{i}$, then $f_{i}^{c}(t_{i})$ yields
player $j$ more than her value $v_{j}$, contradiction.

The CM-condition implies that both players have the same number of types.
Denote this number by $n$. First focus on player $1$. For $t_{1}\in T_{1}$,
write $f_{1}^{c}(t_{1})$ for the vector of conditional probabilities $%
f_{1}^{c}(t_{2}|t_{1})_{t_{2}\in T_{2}}$, and denote by $u_{1}^{\mu }(t_{1})$
the vector $u_{1}^{\mu }(t_1, t_2)_{t_{2}\in T_{2}}$ of payoffs of player $1$%
. Note that with this notation, we have the incentive inequalities are given
by
\begin{align}
f_{1}^{c}(t_{1})\cdot u_{1}^{\mu }(t_{1})& =v_{1}, & & \text{for all }t_{1},
\label{VA1} \\
f_{1}^{c}(t_{1})\cdot u_{1}^{\mu }(t_{1})-f_{1}^{c}(t_{1})\cdot u_{1}^{\mu
}(t_{1}^{\prime })& =0, & & \text{for all }t_{1}\text{ and }t_{1}^{\prime
}\neq t_{1}.  \label{VA2}
\end{align}%
The system \eqref{VA1} and \eqref{VA2} has $n^{2}$ equations and $n^{2}$
variables. This system can be written as
\begin{equation}
f_{1}^{c}(t_{1})\cdot u_{1}^{\mu }(t_{1}^{\prime })=v_{1},\,\text{for all}%
\,\,t_{1}\,\text{and}\,\,t_{1}^{\prime }\neq t_{1}.  \label{VA3}
\end{equation}%
The constraint matrix of \eqref{VA3} is given by
\begin{equation}
M=%
\begin{bmatrix}
\notag f_{1}^{c}(t_{1}^{1}) &  & \mathbf{0} \\
& \ddots &  \\
\mathbf{0} &  & f_{1}^{c}(t_{1}^{1}) \\
& ... &  \\
f_{1}^{c}(t_{1}^{n}) &  & \mathbf{0} \\
& \ddots &  \\
\mathbf{0} &  & f_{1}^{c}(t_{1}^{n})%
\end{bmatrix}%
\end{equation}%
where the elements in $T_{1}$ are indexed by $t_{1}^{1},\dots ,t_{1}^{n}$
and $\mathbf{0}$ denotes the row vector with $n(n-1)$ zeros. Notice that $M$
is an $n^{2}\times n^{2}$ matrix, and that $f$ having full rank implies that
$M$ has linearly independent rows and hence full rank. Consequently, by the
invertible matrix theorem, system \eqref{VA3} has a unique solution given by
$u_{1}^{\mu }(t)=v_{1}$ for all $t $. Similarly we can solve $u_{2}^{\mu
}(t)=v_{2}$ for all $t$. It is easily seen that the set of mechanisms that
implement this utility vector is non-empty.
\end{proof}

\begin{proof}[Proof of Remark 2: A three-player example]
Let $A=\{a_{0},a_{1},a_{2},a_{3}\}$ where each player $i\in \{1,2,3\}$
prefers $a_{i}$ most and $a_{j}$ least (for all $j\neq i$), i.e., $%
u_{i}(a_{i})=1$ and $u_{i}(a_{j})=0$. Each player has two types, $1$ and $2$%
, with disagreement values $u^{1}=0$ and $0<u^{2}<1/3$. Furthermore, for $%
t=(k,l,m)$, write $f(t)=f_{klm}>0$. Suppose that $f$ is symmetric ($%
f_{121}=f_{112}=f_{211}$, $f_{122}=f_{221}=f_{212}$) with, for some $%
\varepsilon \in (0,\frac{1}{12}]$,
\begin{equation*}
f_{111}=\frac{1}{2}-4\varepsilon ,\text{ }f_{222}=\frac{1}{2}-5\varepsilon
,\ f_{211}=\varepsilon ,\ f_{221}=2\varepsilon .
\end{equation*}

Consider the following mechanism $\mu :\{1,2\}^{3}\rightarrow \Delta (A)$:
If all players report the same type, then an equal randomization over $%
A\setminus \{a_{0}\}$ is implemented. If only two players $i$ and $j$ report
the same type, then these players (the majority) are favored and a fair
lottery over $\{a_{i},a_{j}\}$ is implemented. This mechanism is efficient
and type-dependent.
\end{proof}

\begin{proof}[Proof of Theorem 2]
In the main text, we already proved the second part. Here we show that the
first part holds: if $|T_{1}|=2$, $|T_{2}|=l\geq 2$ and $f(t)>0$ for all $t$%
, then an efficient mechanism is type-independent. The proof is by
contradiction. Assume that $\mu $ is a type-dependent and efficient
mechanism. Denote a type of player $1$ (resp. $2$) by $i$ (resp. $j$) and
simplify notation by writing $a_{ij}=u_{1}^{\mu }(i,j)$ and $%
b_{ij}=u_{2}^{\mu }(i,j)$ for the players' expected utilities in the
truthtelling equilibrium when the types are $i$ and $j$. We first consider
the case in which all entries in the first row of player $2$'s payoff matrix
$B$ are different, $b_{1j}\neq b_{1k}$ for all $j\neq k$. Without loss of
generality, we can assume that the columns are ordered such that $%
b_{11}>b_{12}>\dots >b_{1l}$. Since $f(t)>0$ for all $t$, incentive
compatibility implies that we must have $b_{21}<b_{22}<\dots <b_{2l}$, as
otherwise some column $b_{\cdot j}$ would be strictly dominated for player $%
2 $. Furthermore, since $\mu $ is efficient, we must have that the entries
in the first row of player $1$'s payoff matrix $A$ are increasing ($%
a_{11}<a_{12}<\dots <a_{1l}$), while those in the second row are decreasing (%
$a_{21}>a_{22}>\dots >a_{2l}$). Incentive compatibility for player $1$
implies that the payoffs of one row cannot all be larger than the payoffs in
the other row; hence, if $\alpha _{j}=a_{1j}-a_{2j}$, then $\alpha _{j}$ is
increasing in $j$ with $\alpha _{1}<0$ and $\alpha _{l}>0$. For $i=1,2$,
write $f_{i}^{c}$ for the conditional distribution of $j$ given that player $%
1$ has type $i$. The incentive constraints of player 1 are:
\begin{equation}
E(\alpha |f_{1}^{c})\geq 0\text{ and }E(\alpha |f_{2}^{c})\leq 0.  \label{IC}
\end{equation}%
Writing $p_{j}$ for the conditional probability that player $1$ is of type $%
1 $, if player $2$ has type $j$, the incentive constraints for any pair of
types $j<k$ of player 2 are given by
\begin{align*}
b_{1j}p_{j}+b_{2j}(1-p_{j})& \geq b_{1k}p_{j}+b_{2k}(1-p_{j}), \\
b_{1j}p_{k}+b_{2j}(1-p_{k})& \leq b_{1k}p_{k}+b_{2k}(1-p_{k}).
\end{align*}%
Subtracting the second inequality from the first, we get
\begin{equation*}
(b_{1j}-b_{1k}+b_{2k}-b_{2j})(p_{j}-p_{k})\geq 0.
\end{equation*}%
Hence, $p_{j}\geq p_{k}$ since the first term is positive. Hence, the odds
ratio is declining. Consequently, if player $1$ is of type $2$ he is
assigning higher probabilities to the higher values of $j$ as when he is of
type $1$. As $\alpha _{j}$ is increasing in $j$, this implies
\begin{equation*}
E(\alpha |f_{1}^{c})\leq \text{ }E(\alpha |f_{2}^{c}).
\end{equation*}%
Combining this with \eqref{IC} yields
\begin{equation*}
E(\alpha |f_{1}^{c})=\text{ }E(\alpha |f_{2}^{c})=0.
\end{equation*}%
Since $\alpha _{j}$ is increasing in $j$, this is possible only if $%
f_{1j}^{c}=f_{2j}^{c}$ for all $j$. Hence, $f$ is independent. But then, by
Observation 1, $\mu $ must be type-independent. This contradicts the
assumption that we started with. Consequently, if all entries in the first
row of the matrix $B$ are different, then a type-dependent, efficient
mechanism cannot exist.

Now assume that $b_{1j}=b_{1k}$ for some $j,k$. As above, we then must have $%
b_{2j}=b_{2k}$ as otherwise either column $j$ or $k$ would be strictly
dominated for player 2. Furthermore, since all payoff vectors are Pareto
efficient and $\mu $ is also incentive compatible for player $1$, we must
have $a_{1j}=a_{1k}$ and $a_{2j}=a_{2k}$. Hence the columns $j$ and $k$ of
the matrices $A$ and $B$ are identical. Now, consider the reduced problem in
which each set of identical columns is merged, say into a set $J$ and with $%
f_{iJ}=\sum_{j\in J}f_{ij}$ being the probability that state $(i,J)$ is
realized. The mechanism $\mu $ naturally induces a mechanism $\mu ^{\prime }$
on this reduced problem, which is efficient. The first part of the proof
applies to this reduced problem so that $\mu ^{\prime }$ is
type-independent. This implies that all entries in the reduced matrices $%
A^{\prime }$ and $B^{\prime }$ must be the same. But then there can only be
one equivalence class to start with, hence, $\mu $ is type-independent.
\end{proof}

\begin{proof}[Proof of Proposition \protect\ref{Pr:5}]
It suffices to prove the second and the third statements. First assume that
there is 2-sided incomplete information: {{$|T_{1}|>1$ and $|T_{2}|>1$. }}%
(2) Let $f$ be independent. Assume $\mu $ is egalitarian and individually
rational. {\ }Let $\mu _{i}^{0}(t_{i})$ be the probability that $a_{0}$ is
chosen when player $j$ is truthful and player $i$ has type $t_{i}$, and
write $U_{i}^{\mu }(t_{i})=\mu _{i}^{0}(t_{i})t_{i}+r_{i}(t_{i})$.{\ }%
Consider two types $x$ and $y$ of player $i$ with $x<y$. Two incentive
constraints for these types are given by $\mu _{i}^{0}(x)x+r_{i}(x)\geq \mu
_{i}^{0}(y)x+r_{i}(y)$ and $\mu _{i}^{0}(y)y+r_{i}(y)\geq \mu
_{i}^{0}(x)y+r_{i}(x)$, while egalitarianism requires that $\mu
_{i}^{0}(x)x+r_{i}(x)-x=\mu _{i}^{0}(y)y+r_{i}(y)-y$. Substituting this
expression in the incentive constraints and simplifying, we obtain $\mu
_{i}^{0}(y)(y-x)\geq y-x$. Since $y>x$, this implies $\mu _{i}^{0}(y)\geq 1$
. Clearly, as a probability, $\mu _{i}^{0}(y)\leq 1$; hence $\mu
_{i}^{0}(y)=1$. This establishes that $\mu _{i}^{0}(t_{i})=1$ for all $%
t_{i}\neq \underline{t}_{i}$, which implies that $\mu^{0}(t_{i},t_{j})=1 $
for all $t_{j}$, all $t_{i}\neq \underline{t}_{i}$ and $i=1,2$. This, in
turn, implies $\mu ^{0}(t)=1$ for all $t\neq \underline{t}$.

(3) On top of the assumptions from (2), assume $b_{1}=b_{2}=0$. Let $%
t_{i}\neq \underline{t}_{i}$. From (2), we know $\mu _{i}^{0}(t_{i})=1$,
hence elements from $A^{0}$ are selected with probability $0$, and since no
money is available, we have $r_{i}(t_{i})=0$. Consequently, $U_{i}^{\mu
}(t_{i})=t_{i}$ for all $t_{i}\neq \underline{t}_{i}$. Egalitarianism then
implies that also $U_{i}^{\mu }(\underline{t}_{i})=\underline{t}_{i}$. Hence
$U_{i}^{\mu }(t_{i})=t_{i}$ for all $t_{i}$ and $i=1,2$.

If there is 1-sided incomplete information, say {{$|T_{2}|=1$, }}then all
previous statements hold for $i=1$. In (3), egalitarianism implies that $%
U_{2}^{\mu }(\underline{t}_{2})=\underline{t}_{2}$.
\end{proof}

\end{document}